\documentclass[11pt]{article}
\usepackage{graphicx}
\usepackage{amsmath,amsthm,amssymb,amscd,cite,color,enumerate}
\theoremstyle{plain}
\newtheorem{theorem}{Theorem}[section]
\newtheorem{proposition}[theorem]{Proposition}
\newtheorem{lemma}[theorem]{Lemma}
\newtheorem{corollary}[theorem]{Corollary}

\theoremstyle{definition}

\newtheorem{remark}{Remark}[section]

\usepackage{enumerate} 

\newcommand{\I}{\textnormal{I}}
\newcommand{\II}{\textnormal{I\hspace*{-0.4ex}I}}
\newcommand{\III}{\textnormal{I\hspace*{-0.4ex}I\hspace*{-0.4ex}I}}
\newcommand{\Ip}{\textnormal{I\hspace*{-0.3ex}}^\prime}
\newcommand{\IIp}{\textnormal{I\hspace*{-0.4ex}I\hspace*{-0.3ex}}^\prime}

\newcommand{\B}{\mathcal{B}}
\newcommand{\C}{\mathcal{C}}
\newcommand{\D}{\mathcal{D}}

\DeclareMathOperator{\Tor}{Tor}
\DeclareMathOperator{\Res}{Res}
\DeclareMathOperator{\Tr}{Tr}

\begin{document}
    
    \title{Self-Dual Linear Codes over $\mathbb{F}_{q}+u\mathbb{F}_{q}+u^2\mathbb{F}_{q}$   and Their Applications in the Study  of   Quasi-Abelian Codes\thanks{P. Choosuwan was partially supported by the Faculty of Science and Technology, Rajamangala University of Technology Thanyaburi (RMUTT), Pathum Thani, THAILAND. S. Jitman was supported   by the Thailand Research Fund and  Silpakorn University under Research Grant RSA6280042}}
    
    \author{Parinyawat Choosuwan\thanks{P. Choosuwan  is with the Department of Mathematics and Computer Science, Faculty of Science and Technology,
            Rajamangala University of Technology Thanyaburi (RMUTT), Pathum Thani 12110, THAILAND. Email: {parinyawat\_c@rmutt.ac.th}. }  and Somphong  Jitman\thanks{ 
            S. Jitman (Corresponding Author)  is with the  Department of Mathematics, Faculty of Science, Silpakorn University,
            Nakhon Pathom 73000, Thailand.
            Email: {sjitman@gmail.com}.
    }}
    
    \maketitle

 \begin{abstract}  
     Self-dual codes over finite fields and over some finite rings have been of interest and extensively studied due to their nice algebraic structures and wide applications. Recently, characterization and enumeration of Euclidean self-dual linear codes over the ring~$\mathbb{F}_{q}+u\mathbb{F}_{q}+u^2\mathbb{F}_{q}$  with    $u^3=0$ have been  established. 
    In this paper,  Hermitian self-dual linear codes  over  $\mathbb{F}_{q}+u\mathbb{F}_{q}+u^2\mathbb{F}_{q}$ are   studied  for  all  square prime powers~$q$. 
    Complete  characterization and enumeration of such codes are given.    
    Subsequently,  algebraic characterization  of $H$-quasi-abelian codes  in $\mathbb{F}_q[G]$ is studied, where  $H\leq G$ are finite abelian groups and  $\mathbb{F}_q[H]$ is a principal ideal group algebra. General characterization and enumeration of  $H$-quasi-abelian codes and self-dual $H$-quasi-abelian  codes  in $\mathbb{F}_q[G]$   are given. 
    For  the special case where the field characteristic is $3$,   an explicit formula for the number    of    self-dual $A\times \mathbb{Z}_3$-quasi-abelian codes   in $\mathbb{F}_{3^m}[A\times \mathbb{Z}_3\times B]$   is  determined  for all finite abelian groups $A$ and $B$ such that $3\nmid |A|$  as well as their construction.  Precisely, such codes  can be  represented  in terms of  linear codes and  self-dual linear  codes over     $\mathbb{F}_{3^m}+u\mathbb{F}_{3^m}+u^2\mathbb{F}_{3^m}$.
    \\
    \textbf{2010 AMS Classification.}  94B15, 94B05, 94B60
\end{abstract}


\section{Introduction}

Self-dual  linear codes over finite fields  form an  interesting class of  linear codes that have been extensively  studied   due to their nice  algebraic structures and wide applications (see \cite{JLX2011}, \cite{JLLX2012},\cite{JLS2013}, \cite{RS1998} and references therein).  
Codes over finite rings have been of interest  after   it was shown that some binary non-linear codes such as the Kerdock, Preparata  and Goethal codes are the Gray images of linear codes over $\mathbb{Z}_4$ in \cite{HKCSS1994}. In general,  families of linear  codes and self-dual linear  codes over finite chain rings are now become of interest.  In \cite{NNW2009},  the mass formula for Euclidean self-dual linear codes over $\mathbb{Z}_{p^3}$ has been studied. Characterization and enumeration of    Euclidean self-dual linear codes over the ring $\mathbb{F}_{q}+u\mathbb{F}_{q}+u^2\mathbb{F}_{q}$  with $u^3=0$   have been given in \cite{BNV2018}.

Algebraically  structured codes  over finite fields such as cyclic codes, abelian codes  and quasi-abelian codes      are another  important family of linear codes  that have been extensively studied for both  theoretical and practical reasons  (see  \cite{SD1967}, \cite{JLX2011},    \cite{JL2015}, \cite{JLLX2012},   \cite{JLS2013}     and references therein).
In \cite{JL2015},     $H$-quasi-abelian codes and self-dual   $H$-quasi-abelian codes  in $\mathbb{F}_q [G]$   have been studied in the case where $\mathbb{F}_q[H]$ is semisimple 

To the best of our knowledge, Hermitian self-dual linear codes over   $\mathbb{F}_{q}+u\mathbb{F}_{q}+u^2\mathbb{F}_{q}$   and non-semisimple $H$-quasi-abelian codes  in $\mathbb{F}_q [G]$   have not been well studied. 
The  goals of this paper are to investigate  the following families of linear codes and their links.  1) Hermitian self-dual linear codes over $\mathbb{F}_{q}+u\mathbb{F}_{q}+u^2\mathbb{F}_{q}$ where   $q$ is a square prime power. 2)  $H$-quasi-abelian codes and self-dual $H$-quasi-abelian codes in group algebras $\mathbb{F}_q[G]$, where $H\leq G$ are finite abelian groups and  $\mathbb{F}_q[H]$ is a principal ideal group algebra.  
In Section $2$,  some results on linear codes and  Euclidean self-dual linear codes over $\mathbb{F}_{q}+u\mathbb{F}_{q}+u^2\mathbb{F}_{q}$ are recalled. In Section $3$, characterization and enumeration Hermitian self-dual linear codes  of length~$n$ over $\mathbb{F}_{q}+u\mathbb{F}_{q}+u^2\mathbb{F}_{q}$ are established  for all square prime powers $q$.   In Section 4, the study of $H$-quasi-abelian codes  in  $\mathbb{F}_q[G]$ is given, where $\mathbb{F}_q[H]$ is a principal ideal group algebra. 
In the special case where the field characteristic is $3$,   the characterization and enumeration of  $A\times \mathbb{Z}_3$-quasi-abelian codes   and self-dual $A\times \mathbb{Z}_3$-quasi-abelian codes in $\mathbb{F}_{3^m}[A\times \mathbb{Z}_3\times B]$  are completely determined  in terms of  linear and self-dual linear codes over     $\mathbb{F}_{3^m}+u\mathbb{F}_{3^m}+u^2\mathbb{F}_{3^m}$ obtained in Section $3$ for all finite abelian groups $A$ and $B$ such that $3\nmid |A|$. Summary and remarks are given in Section $5$.

\section{Preliminaries}\label{sec:pre}
In this section,  basic results on linear codes and Euclidean   self-dual linear codes over rings are recalled.

\subsection{Linear Codes over $\mathbb{F}_{q}+u\mathbb{F}_{q}+\dots +u^{e-1}\mathbb{F}_{q}$}

For a prime power $q$, denote by $\mathbb{F}_{q}$ the finite field of order $q$.  Let $\mathbb{F}_{q}+u\mathbb{F}_{q}+\dots +u^{e-1}\mathbb{F}_{q}:=\{a_0+ua_1+\dots+u^{e-1}a_{e-1}\mid a_i\in \mathbb{F}_{q} \text{ for all } 0\leq i<e\}$ be a ring, where the addition and  multiplication are  defined  as in the usual polynomial ring over $\mathbb{F}_{q}$ with indeterminate $u$ together with  the condition $u^e=0$. It is easily seen that   $\mathbb{F}_{q}+u\mathbb{F}_{q}+\dots+u^{e-1}\mathbb{F}_{q}$  is isomorphic to $\mathbb{F}_{q}[u]/\langle u^e \rangle$ as rings.  
The  {\em Galois extension} of $\mathbb{F}_{q}+u\mathbb{F}_{q}+\dots+u^{e-1}\mathbb{F}_{q}$  of degree $m$ is defined to be the quotient ring $(\mathbb{F}_{q}+u\mathbb{F}_{q}+\dots+u^{e-1}\mathbb{F}_{q})[x]/\langle f(x)\rangle$, where $f(x)$ is an irreducible polynomial of degree $m$ over $\mathbb{F}_{q}$. It is not difficult to see that  the Galois extension of $\mathbb{F}_{q}+u\mathbb{F}_{q}+\dots +u^{e-1}\mathbb{F}_{q}$  of degree $m$ is isomorphic to  $\mathbb{F}_{q^{m}}+u\mathbb{F}_{q^{m}}+\dots+u^{e-1}\mathbb{F}_{q^{m}}$.  The ring  $\mathbb{F}_{q}+u\mathbb{F}_{q}+\dots +u^{e-1}\mathbb{F}_{q}$ is a finite chain ring with maximal ideal  $\langle u\rangle $, nilpotency index $e$  and residue field $\mathbb{F}_q$.
In addition, if $q$ is a square, the mapping    $\bar{~ }:\mathbb{F}_q\to \mathbb{F}_q$ defined by $a\mapsto a^{\sqrt{q}} $ is a field  automorphism  on $\mathbb{F}_{q}$ of order $2$. Extend  $\bar{~ }$  to be a ring automorphism of order $2$ on $\mathbb{F}_{q}+u\mathbb{F}_{q}+\dots +u^{e-1}\mathbb{F}_{q}$ of the form 
$\overline{a_0+ua_1+\dots+u^{e-1}a_{e-1}}= \overline{a_0} +u\overline{a_1} +\dots+u^{e-1}\overline{a_{e-1} }$.

Let $n$ be  a positive integer  and  let $R$ be a finite ring. The {\em Euclidean inner product} of $\boldsymbol{u}=(u_{0},u_{1},\ldots ,u_{n-1})$ and $\boldsymbol{v}= (v_{0},v_{1},\ldots ,v_{n-1})$ in $R^n$ is defined to be
\begin{align*}
\langle \boldsymbol{u},\boldsymbol{v}\rangle_{\rm E}:=\displaystyle\sum_{i=0}^{n-1} u_iv_i. 
\end{align*}
In the case where $q$ is a square and $R\in \{ \mathbb{F}_q,\mathbb{F}_{q}+u\mathbb{F}_{q}+\dots +u^{e-1}\mathbb{F}_{q}\}$, the {\em Hermitian inner product} of $\boldsymbol{u}$ and $\boldsymbol{v}$ in $R^n$   is defined to be
\begin{align*}
\langle \boldsymbol{u},\boldsymbol{v}\rangle_{\rm H}:=\displaystyle\sum_{i=0}^{n-1} u_i\overline{v_i}.
\end{align*}

A {\em linear code} $\C$ of length $n$ over the ring~$R$ is defined to be  an $R$-submodule of the $R$-module~$R^n$. A matrix $G$ over $R$ is called a {\em generator matrix} for~$C$ if the rows of $G$ generate all the elements of $\C$ and none of the rows can be written as a linear combination of the others. Linear codes~$\C_1$ and $\C_2$  over $R$ are said to be {\em equivalent} if there exists a monomial matrix $P$ such that $\C_2 = \C_1 P:= \{\boldsymbol{c} P \mid  \boldsymbol{c} \in \C_1\}$.  Denote by 
$\C^{\perp_{\rm E}}=\{ \boldsymbol{v}\in R^n \mid  \langle \boldsymbol{u},\boldsymbol{v}\rangle_{\rm E}=0 \}$ and $\C^{\perp_{\rm H}}=\{ \boldsymbol{v}\in R^n \mid  \langle \boldsymbol{u},\boldsymbol{v}\rangle_{\rm H}=0 \}$ the {\em Euclidean} and {\em Hermitian duals} of $\C$, respectively.  A linear  code $\C$ is said to be {\em Euclidean} (resp., {\em Hermitian}) {\em self-orthogonal} if $\C\subseteq \C^{\perp_{\rm E}}$ (resp., $\C \subseteq \C^{\perp_{\rm H}}$). It is called  {\em Euclidean} (resp., {\em Hermitian}) {\em self-dual} if $\C=\C^{\perp_{\rm E}}$ (resp., $\C=\C^{\perp_{\rm H}}$).

In Section $3$ and the remaining parts of this section, we focus on linear and self-dual  linear codes over $\mathbb{F}_{q}+u\mathbb{F}_{q}+u^2\mathbb{F}_{q}$.
In \cite{NH2000}, it has been shown that every linear code of length $n$ over $\mathbb{F}_{q}+u\mathbb{F}_{q}+u^2\mathbb{F}_{q}$ is permutation equivalent to a code $\C$ with generator matrix
\begin{equation}\label{gmlinearcode}
G=\begin{bmatrix}I_k	&A_2  & A_3  & A_4\\ 
0	&  uI_l& uB_3   & uB_4  \\ 
0	& 0 & u^2I_m & u^2C_4 
\end{bmatrix}
=\begin{bmatrix}A^{\prime}\\ 
uB^{\prime}	  \\ 
u^2C
\end{bmatrix},
\end{equation}
where $I_r$ is the identity matrix of order $r$, $A_3 = A_{30} + uA_{31}, B_4 = B_{40} + uB_{41}, A_4 = A_{40} +uA_{41} +u^2A_{42}$, and $A_2, B_3,C_4$, $A_{ij}$ and $B_{ij}$ are matrices of appropriate sizes over~$\mathbb{F}_q$. In this case, the code $\C$ is said to be of type $\{k, l, m\}$ and it contains $q^{3k+2l+m}$ codewords.

For each linear code $\C$ of~length $n$ over $\mathbb{F}_{q}+u\mathbb{F}_{q}+u^2\mathbb{F}_{q}$ and $i\in\{0,1,2\}$,   the {\em ${i}${\rm th} torsion code} of $\C$ is a linear code of length $n$ over $\mathbb{F}_q$ defined to be 
\begin{align*}
\Tor_i(\C)= \displaystyle\left\{ \boldsymbol{v}(\bmod u) \mid \boldsymbol{v}\in \left(\mathbb{F}_{q}+u\mathbb{F}_{q}+u^2\mathbb{F}_{q}\right)^n \text{ and } u^i\boldsymbol{v}\in \C\right\}.
\end{align*}
The code $\Tor_0({\C})$ is sometime called    the {\em residue code} of  $\C$ and denoted it  by $\Res(\C)$. From the definitions, it is obvious that $\Res(\C)=\Tor_0(\C)\subseteq \Tor_1(\C) \subseteq \Tor_2(\C)$.
For a linear code $\C$ of length $n$ over $\mathbb{F}_{q}+u\mathbb{F}_{q}+u^2\mathbb{F}_{q}$ with generator matrix $G$ given  in \eqref{gmlinearcode},  the residue code $\Res(\C)$ has dimension $k$ and generator matrix
􏰋􏰌\begin{equation}\label{gmres}
G=\begin{bmatrix}I_k	&A_2  & A_{30}  & A_{40}
\end{bmatrix},
\end{equation}
the first torsion code $\Tor_1(\C)$ has dimension $k + l$ and generator matrix
􏰍􏰎􏰍 􏰎\begin{equation}\label{gmtor1}
\begin{bmatrix}A\\ 
B
\end{bmatrix}
=\begin{bmatrix}I_k	&A_2  & A_{30}  & A_{40}\\ 
0	&  I_l& B_3   & B_{4  0}
\end{bmatrix},
\end{equation}
and the second torsion code $\Tor_2(\C)$ has dimension $k + l + m$ and generator matrix
\begin{equation}\label{gmtor2}
\begin{bmatrix}A\\ 
B\\ 
C
\end{bmatrix}
=\begin{bmatrix}I_k	&A_2  & A_{30}  & A_{40}\\ 
0	&  I_l& B_3   & B_{40}  \\ 
0	& 0 & I_m & C_4 
\end{bmatrix}.
\end{equation}

For $0\leq k \leq n$, the {\em Gaussian coefficient}
is defined  to be 
\[ \begin{bmatrix}n\\ 
k
\end{bmatrix}_q=\displaystyle\dfrac{(q^n-1)(q^n-q)\cdots (q^n-q^{k-1})}{(q^k-1)(q^k-q)\cdots (q^k-q^{k-1})}.\]

Let $N_e(q,n)$ denote the number of distinct  linear codes of   length $n$ over $\mathbb{F}_{q}+u\mathbb{F}_{q}+\dots+u^{e-1}\mathbb{F}_{q}$.  The number $N_e(q,n)$ has been studied and summarized in \cite{CG2015}. For $e=3$, we have the following result. 
\begin{proposition}[{\cite[Lemma 2.2]{CG2015}}] \label{N3} Let $q$ be a prime power and let $n$ be a positive integer. Then\[N_3(q,n)=1+\sum_{t=1}^{3}~ \sum_{n\geq h_1\geq h_2 \geq \dots\geq h_t >h_{t+1}=0} ~\prod_{j=1}^{t}   \begin{bmatrix}n-h_{j+1}\\ 
    h_j-h_{j+1}
    \end{bmatrix}_q q^{h_{j+1}(n-h_j)}.\]
    
\end{proposition}

\subsection{Euclidean Self-Dual Linear Codes over $\mathbb{F}_{q}+u\mathbb{F}_{q}+u^2\mathbb{F}_{q}$}
Let $\C$ be a linear code of length $n$ and type $\{k, l, m\}$ over $\mathbb{F}_{q}+u\mathbb{F}_{q}+u^2\mathbb{F}_{q}$  and let $h = n- (k + l + m)$. In \cite{BNV2018}, it has been shown that the  Euclidean dual  $\C^{\perp_{\rm E}}$ of $\C$  is of type $\{h, m, l\}$ and it contains  $q^{3h+2m+l}$ codewords. Therefore,   $k=h$ and $l=m$ whenever $\C$ is Euclidean self-dual. Consequently, every  Euclidean self-dual code of  type $\{k, l, m\}$  over $\mathbb{F}_{q}+u\mathbb{F}_{q}+u^2\mathbb{F}_{q}$  has  even length $n = 2(k + l)$.

Characterization of Euclidean self-dual linear codes of even length $n$ over $\mathbb{F}_{q}+u\mathbb{F}_{q}+u^2\mathbb{F}_{q}$ has been established in \cite{BNV2018}.

\begin{proposition}[{\cite[Proposition 1]{BNV2018}}]\label{charesd}Let $q$ be a prime power and let  $\C$ be a linear code of  length $n$ and type $\{k, l, m\}$  over $\mathbb{F}_{q}+u\mathbb{F}_{q}+u^2\mathbb{F}_{q}$  with  generator matrix $G$ in the form of  \eqref{gmlinearcode}. Then $\C$ is Euclidean self-dual if and only if $k=h,l=m$ and the following  conditions hold:
    \begin{align}
    A^{\prime}{A^{\prime}}^T&\equiv 0 \pmod {u^3},\\
    A^{\prime}{B^{\prime}}^T&\equiv 0 \pmod {u^2},\\
    B^{\prime}{B^{\prime}}^T&\equiv 0 \pmod u,\\
    A^{\prime}C^T&\equiv 0 \pmod u.
    \end{align}
\end{proposition}

For a positive integer $n$ and a prime power $q$, let $\sigma_{\rm E}(q,n)$  denote the number of  Euclidean self-dual  linear codes of length $n$ over $\mathbb{F}_q$. Further, if $q$ is a square prime power,  let $\sigma_{\rm H}(q,n)$  denote the number of  Hermitian  self-dual  linear codes of length $n$ over $\mathbb{F}_q$.  The following results in {\cite{PV1968} and \cite{RS1998}} are useful in the enumeration of self-dual linear codes over  $\mathbb{F}_{q}+u\mathbb{F}_{q}+u^2\mathbb{F}_{q}$.

\begin{lemma}[{\cite{PV1968} and \cite{RS1998}}] Let $q$ be a  prime power and let  $n$  be a positive integer. Then
    
    \begin{align}
    \sigma_{\rm E}(q,l)&=
    \begin{cases} 
    \displaystyle\prod_{i=1}^{\frac{n}{2}-1}(q^{i}+1) & \text{if $q$ and $n$ are even},\\
    \displaystyle2\prod_{i=1}^{\frac{n}{2}-1}(q^{i}+1) & \text{if $q\equiv 1\,({\rm mod}\, 4)$ and $2\mid n$,}\\ 
    \displaystyle2\prod_{i=1}^{\frac{n}{2}-1}(q^{i}+1) & \text{if $q\equiv 3\,({\rm mod}\, 4)$  and $4\mid n$,}\\
    0		& \text{otherwise}.						
    \end{cases}\label{Edual}
    \end{align}
    If $q$ is square, then
    \begin{align}
    \sigma_{\rm H}(q,n)&=\begin{cases}\displaystyle\prod_{i=0}^{\frac{n}{2}-1}(q^{i+\frac{1}{2}}+1) & \text{if $n$ is even},\\	
    0& \text{otherwise}.\label{Hdual}
    \end{cases}
    \end{align}
    The empty product is regarded as $1$.
\end{lemma}

Let $NE_e(q,n)$ denote the number of distinct Euclidean self-dual linear codes of   length $n$ over $\mathbb{F}_{q}+u\mathbb{F}_{q}+\dots+u^{e-1}\mathbb{F}_{q}$. The  value of $NE_3(q,n)$  has been  completely determined  in  \cite{BNV2018}.
\begin{theorem}[{\cite[Theorem 1]{BNV2018}}] Let $q$ be a prime power and let $n$ be a positive integer. Then
    \[NE_3(q,n)=
    \begin{cases}
    \sigma_{\rm E}\displaystyle(q,n)\displaystyle\sum_{k=0}^{n/2}\begin{bmatrix}\displaystyle\dfrac{n}{2}\\ 
    k
    \end{bmatrix}_qq^{kn/2} &\text{ if } q \text{ is even and } n \text{ is even},\\
    \sigma_{\rm E}\displaystyle(q,n)\displaystyle\sum_{k=0}^{n/2}\begin{bmatrix}\displaystyle\dfrac{n}{2}\\ 
    k
    \end{bmatrix}_qq^{k(n/2-1)} &\text{ if } q \text{ is odd and } n \text{ is even},\\
    0& \text{ otherwise}.
    \end{cases}\]
\end{theorem}

\section{Hermitian Self-Dual Linear Codes over $\mathbb{F}_{q}+u\mathbb{F}_{q}+u^2\mathbb{F}_{q}$}
\label{sec:3}
In this section, we focus on characterization and enumeration of Hermitian  self-dual linear codes  of length $n$  over~${\mathbb{F}_{q}+u\mathbb{F}_{q}+u^2\mathbb{F}_{q}}$.

Throughout this section, we assume that $q$ is a square prime power.  For  each positive integer $n$, let 
$NH_e(q,n)$ denote the number of distinct Hermitian self-dual linear codes of  length $n$ over $\mathbb{F}_{q}+u\mathbb{F}_{q}+\dots +u^{e-1}\mathbb{F}_{q}$.   By extending  techniques introduced  in  \cite{BNV2018}, the characterization and the number  $NH_3(q,n)$ of Hermitian  self-dual linear codes   of length $n$ over~${\mathbb{F}_{q}+u\mathbb{F}_{q}+u^2\mathbb{F}_{q}}$ are established.

Let $\C$ be a linear code of length~$n$ over~${\mathbb{F}_{q}+u\mathbb{F}_{q}+u^2\mathbb{F}_{q}}$  of  type $\{k, l, m\}$ and let   $h = n- (k + l + m)$. Using argument similar to those in Section $2$ of \cite{BNV2018}, it can be deduced that  the Hermitian dual  $\C^{\perp_{\rm H}}$ of $\C$  is of type $\{h, m, l\}$ and it contains $q^{3h+2m+l}$ codewords. It follows that   $k=h$ and $l=m$  if $\C$ is   Hermitian self-dual.   Hence,  every  Hermitian self-dual code of  type $\{k, l, m\}$  over~${\mathbb{F}_{q}+u\mathbb{F}_{q}+u^2\mathbb{F}_{q}}$  has  even length $n = 2(k + l)$.

For a matrix $A=[a_{ij}]_{s\times t}$ over $\mathbb{F}_{q}+u\mathbb{F}_{q}+u^2\mathbb{F}_{q}$,  let $\overline{A}:=\left[\overline{a_{ij}}\right]_{s\times t}$ and  $A^{\dagger}:= \overline{A}^T$. Characterization of Hermitian self-dual linear codes of even length $n$ over $\mathbb{F}_{q}+u\mathbb{F}_{q}+u^2\mathbb{F}_{q}$ is given in the following theorem. 

\begin{proposition}\label{charhsd} Let $q$ be a square prime power and let  $\C$ be a linear code of  even length $n$ and type $\{k, l, m\}$  over $\mathbb{F}_{q}+u\mathbb{F}_{q}+u^2\mathbb{F}_{q}$ with  generator matrix $G$ in the form of  \eqref{gmlinearcode}. Then $\C$ is Hermitian self-dual if and only if $k=h,$ $l=m$ and the following hold:
    \begin{align}
    A^{\prime}{A^{\prime}}^{\dagger}&\equiv 0 \pmod {u^3}\label{char1},\\
    A^{\prime}{B^{\prime}}^{\dagger}&\equiv 0 \pmod {u^2}\label{char2},\\
    B^{\prime}{B^{\prime}}^{\dagger}&\equiv 0 \pmod u\label{char3},\\
    A^{\prime}C^{\dagger}&\equiv 0 \pmod u.\label{char4}
    \end{align}
\end{proposition}
\begin{proof}Assume that $\C$ is  Hermitian self-dual.  By the above discussion, we have $k=h,$ $l=m$ and 
    \[\begin{bmatrix}A^{\prime}\\ 
    uB^{\prime}\\ 
    u^2C
    \end{bmatrix}\begin{bmatrix}A^{\prime}\\ 
    uB^{\prime}\\ 
    u^2C
    \end{bmatrix}^{\dagger}\equiv[0]\pmod {u^3}\]
    which  are  equivalent to the  conditions (\ref{char1})--(\ref{char4}).
    
    Conversely, assume that $\C$ is a linear code such that $k=h,l=m$ and the conditions (\ref{char1})--(\ref{char4}) hold. From   (\ref{char1})--(\ref{char4}), it is not difficult to see  that $\C$ is Hermitian self-orthogonal. Equivalently,   $\C\subseteq \C^{\perp_{\rm H}}$. Since $k=h$ and $l=m$, we have $|\C| =| \C^{\perp_{\rm H}}| $  which implies  that  $\C =\C^{\perp_{\rm H}}$. Therefore, $\C$ is Hermitian self-dual as desired.
\end{proof}

\begin{corollary} \label{charHSD} Let $\C$ be a  Hermitian self-dual linear code of length $n$ over $\mathbb{F}_{q}+u\mathbb{F}_{q}+u^2\mathbb{F}_{q}$. Then   the following statements holds. 
    \begin{enumerate}[~~~1)]
        \item $\Tor_1(\C)$ is a Hermitian self-dual code of length $n$ over $\mathbb{F}_q$.
        \item $\Tor_2(\C) = \Res(\C)^{\perp_{\rm H}}$.
    \end{enumerate}
\end{corollary}
\begin{proof}
    Assume that $C$ is  of   type $\{k, l, m\}$   over   $\mathbb{F}_{q}+u\mathbb{F}_{q}+u^2\mathbb{F}_{q}$.  
Then 
    From (\ref{char1})--(\ref{char3}), it follows that $\Tor_1(\C)$ is  Hermitian self-orthogonal.  Since $\dim (\Tor_1(\C))=k + l=\frac{n}{2}=\dim ((\Tor_1(\C))^{\perp_{\rm H}})$,  $\Tor_1(\C)$ is  Hermitian self-dual. From  (\ref{char1})--(\ref{char4}),  we have  $\Tor_2(\C) \subseteq \Res(\C)^{\perp_{\rm H}}$. Since $\dim (\Tor_2(\C))=k + 2l=n-k=\dim ((\Res(\C))^{\perp_{\rm H}})$,  we have  $\Tor_2(\C) = \Res(\C)^{\perp_{\rm H}}$.
%
%
%
%
%
\end{proof}
Since $\Res(\C)=\Tor_0(\C)\subseteq \Tor_1(\C) \subseteq \Tor_2(\C)$, it can be concluded further that $\Res(\C)$ is Hermitian self-orthogonal for all Hermitian self-dual linear codes $\C  $  over $\mathbb{F}_{q}+u\mathbb{F}_{q}+u^2\mathbb{F}_{q}$.

From Corollary \ref{charHSD},     a  Hermitian self-dual code $\C$ of length $n$ over $\mathbb{F}_{q}+u\mathbb{F}_{q}+u^2\mathbb{F}_{q}$  can be  induced by  Hermitian self-dual linear codes of length $n$ over $\mathbb{F}_{q}$.  For a given Hermitian self-dual code $\C_1$ of length $n$ over $\mathbb{F}_q$, we first aim to determine the number of Hermitian self-dual linear  codes $\C$ of length $n$  over $\mathbb{F}_{q}+u\mathbb{F}_{q}+u^2\mathbb{F}_{q}$  such that  $\Tor_1(\C) = \C_1$. 

\begin{proposition} \label{HSDC1} Let $q$ be a square prime power and let $n$ be an even positive integer.  Let $\C_1$ be a Hermitian self-dual linear code of length $n$ over $\mathbb{F}_q$.  Then, for each $0\leq k\leq \frac{n}{2}$, there are $q^{\frac{kn}{2}}$ Hermitian self-dual linear codes of length $n$  over $\mathbb{F}_{q}+u\mathbb{F}_{q}+u^2\mathbb{F}_{q}$   corresponding to each subspace of $\C_1$ of dimension $k$.
\end{proposition}

\begin{proof}
    Let $\C_1$ be a Hermitian self-dual linear code of length $n$ over $\mathbb{F}_q$ with dimension $\frac{n}{2}=k + l$ and generator matrix
    􏰍􏰎􏰍 􏰎\begin{equation}
    \begin{bmatrix}A\\ 
    B
    \end{bmatrix}
    =\begin{bmatrix}I_k	&A_2  & A_{30}  & A_{40}\\ 
    0	&  I_l& B_3   & B_{4  0}
    \end{bmatrix},
    \end{equation}
    where the columns are grouped into blocks of sizes $k,l,l$ and $k$. 
    
    Since  $\C_1$ is Hermitian self-dual, we have 
    \begin{align}
    I_k+A_2A_2^{\dagger}+A_{30}A_{30}^{\dagger}+A_{40}A_{40}^{\dagger}&=0\label{char5},\\
    A_2+A_{30}B_3^{\dagger}+A_{40}B_{40}^{\dagger}&=0\label{char6},\\
    I_l+B_3B_3^\dagger+B_{40}B_{40}^\dagger&=0.\label{char7}
    \end{align}
    Let $H=\begin{bmatrix}\overline{A_{30}}	&\overline{A_{40} } \\ 
    \overline{B_3}	&  \overline{B_{40}}
    \end{bmatrix}$.  
    From (\ref{char5})--(\ref{char7}), it can be deduced that 
    \begin{align*}
    H(-H^{\dagger})&=-HH^{\dagger}\\
    &=-H\overline{H}^T\\
    &=\begin{bmatrix}-\overline{A_{30}}A_{30}^{T}-\overline{A_{40}}A_{40}^{T}	&-\overline{A_{30}}B_3^{T}-\overline{A_{40}}B_{40}^{T}  \\ 
    \left( -\overline{A_{30}}B_3^{T}-\overline{A_{40}}B_{40}^{T}\right)^{T}&  -\overline{B_3}B_3^{T}-\overline{B_{40}}B_{40}^{T}
    \end{bmatrix}\\
    &=\begin{bmatrix}I_k+\overline{A_2}A_2^{T}	&\overline{A_{2}}  \\ 
    A_{2}^{T}&  I_l
    \end{bmatrix}.
    \end{align*}
    Let $J=\begin{bmatrix}I_k	&-\overline{A_{2}}  \\ 
    -A_{2}^{T}	&  I_l+A_2^{T}\overline{A_2}
    \end{bmatrix}$. Then 
    $H(-H^{\dagger})J=\begin{bmatrix}I_k	&0 \\ 
    0	&  I_l
    \end{bmatrix}$ which implies that  $H$ is invertible.
    
    Let $\C_0$ be a  $k$-dimensional $\mathbb{F}_q$-subspace of $\C_1$  with generator matrix $A$. Since $\C_1$ is Hermitian self-dual, it follows  that $\C_0$ is Hermitian  self-orthogonal. Up to permutation of the last $(k + l)$ columns (if necessary), its follows that $C_0^{\perp_{\rm H}}$ has a generator matrix  of the form 
    \begin{equation}
    \begin{bmatrix}I_k	&A_2  & A_{30}  & A_{40}\\ 
    0	&  I_l& B_3   & B_{40}  \\ 
    0	& 0 & I_l & C_4 
    \end{bmatrix}. 
    \end{equation}
    Then $A_{30}=-A_{40}C_4^{\dagger}$  which implies that 
    $
    H=\begin{bmatrix}-\overline{A_{40}}C_4^T&\overline{A_{40} } \\ 
    \overline{B_{3}} 	&  \overline{B_{40}}
    \end{bmatrix}$. Since $H$ is invertible, it follows that 
    $A_{40}$ is invertible.
    
    Next, we determined the  matrices over $\mathbb{F}_{q}+u\mathbb{F}_{q}+u^2\mathbb{F}_{q}$ satisfying conditions (\ref{char1})--(\ref{char4}) which are equivalent
    to
    \begin{align}
    I_k+A_2A_2^{\dagger}+A_{3}A_{3}^{\dagger}+A_{4}A_{4}^{\dagger}&\equiv 0 \pmod {u^3}\label{char8}\\
    A_2+A_{3}B_3^{\dagger}+A_{4}B_{4}^{\dagger}&\equiv 0 \pmod {u^2}\label{char9}\\
    I_l+B_3B_3^{\dagger}+B_{4}B_{4}^{\dagger}&\equiv 0 \pmod {u}\label{char10}\\
    A_3+A_4C_4^{\dagger}&\equiv 0 \pmod {u}.\label{char11}
    \end{align}
    The matrices $A_2, B_3$ and $C_4$ are considered modulo $u$, i.e. all the entries in $A_2, B_3$ and $C_4$ are in $\mathbb{F}_q$. The matrices $A_3$ and $B_4$ are considered modulo $u^2$ while $A_4$ is considered modulo $u^3$. From these fact, let $A_3 = A_{30} +uA_{31}, B_4 = B_{40} +uB_{41}$ and $A_4 = A_{40} +uA_{41} +u^2A_{42}$, where $A_{31}, B_{41}, A_{41}$ and $A_{42}$ are matrices of appropriate sizes over $\mathbb{F}_q$. Therefore, we can write (\ref{char8}) as
    \begin{align*}
    \left(I_k + A_2A_2^{\dagger} + A_{30}A_{30}^{\dagger} + A_{40}A_{40}^{\dagger}\right) +u\left(\widetilde{A_{30}A_{31}^{\dagger}}+ \widetilde{A_{40}A_{41}^{\dagger}}\right)&\\
    +u^2\left(A_{31}A_{31}^{\dagger}+A_{41}A_{41}^{\dagger}+\widetilde{A_{40}A_{42}^{\dagger}}\right)&\equiv 0 \pmod {u^3},
    \end{align*}
    where $\widetilde{X}:=X+X^{\dagger}$. We can also write (\ref{char9}) as
    \[\left(A_2+A_{30}B_3^{\dagger}+A_{40}B_{40}^{\dagger}\right) +u\left(A_{31}B_3^{\dagger}+A_{40}B_{41}^{\dagger}+A_{41}B_{40}^{\dagger}\right)\equiv 0 \pmod {u^2} \]
    By substituting (\ref{char7}) into (\ref{char9}), we obtain that
    \[B_{41}^{\dagger}=-A_{40}^{-1}\left(A_{31}B_3^{\dagger}+A_{41}B_{40}^{\dagger}\right),\]
    From \eqref{char11}, $C_4$ is uniquely determined as
    \[C_4=\left(-A_{40}^{-1}A_{30}\right)^{\dagger}.\]
    It is sufficient  to focus on  (\ref{char8}) because (\ref{char7}) is the same as (\ref{char10}). From (\ref{char5}), we have  to determined  the matrices satisfying the following:
    \begin{align}
    \widetilde{A_{30}A_{31}^{\dagger}}+\widetilde{A_{40}A_{41}^{\dagger}}&=0,\label{charh1}\\
    A_{31}A_{31}^{\dagger}+A_{41}A_{41}^{\dagger}+\widetilde{A_{40}A_{42}^{\dagger}}&=0.\label{charh2}
    \end{align}
    Hence,  $\C$ is a Hermitian self-dual linear code if and only if conditions (\ref{charh1}) and (\ref{charh2}) are satisfied.
    
    Therefore, the number of Hermitian self-dual linear codes of length $n$ over  $\mathbb{F}_{q}+u\mathbb{F}_{q}+u^2\mathbb{F}_{q}$ whose the $1$st torsion is $\C_1$   equals  the number of solutions of the system of matrix equations (\ref{charh1}) and (\ref{charh2}).
    
    We take an arbitrary matrix $A_{31} \in M_{k\times l}(\mathbb{F}_q)$ and put $[g_{ij}] = \widetilde{A_{30}A_{31}^\dagger}$ and $[x_{ij} ] = A_{40} A_{41}^\dagger$. Then condition (\ref{charh1}) is equivalent to
    \[g_{ij} +x_{ij} +\overline{x_{ji}} =0.\]
    Then $g_{ii} = x_{ii} +\overline{x_{ii}}  = \Tr(x_{ii}) \in \mathbb{F}_{\sqrt{q}}$ for each $1\leq i \leq k$, where $\Tr :\mathbb{F}_{q}\rightarrow \mathbb{F}_{\sqrt{q}}$ is the {\em trace map}  defined by $\alpha \mapsto \overline{ \alpha }+ \alpha $ for all $\alpha \in \mathbb{F}_{q}$. Note that    $|\Tr^{-1}(a)|=\sqrt{q} $ for all $a\in \mathbb{F}_{\sqrt{q}}$.  Then we have $ x_{ii} \in \Tr^{-1}(g_{ii})$ for all $1\leq i \leq k$,   $ x_{ji} \in  \mathbb{F}_q$ and $ x_{ij} = -g_{ij}-x_{ij} $ for each $1\leq i<j \leq k$.
    Therefore  $A_{41} = (A_{40}^{-1}[x_{ij}])^{\dagger}$ . Thus we have $q^{kl}$ possible choices for $A_{31}$  and $q^{\frac{k(k-1)}{2}+\frac{k}{2}}=q^{\frac{k^2}{2}}$ for $A_{41}$.

    For  fixed matrices $A_{31}$ and $A_{41}$, let  $[h_{ij}]=A_{31}A_{31}^{\dagger} +A_{41}A_{41}^{\dagger}$ and $[y_{ij}]=A_{40}A_{42}^\dagger$. Then (\ref{charh2}) is equivalent to
    \[h_{ij} +y_{ij} +\overline{y_{ji} }=0.\]
    Using a similar argument as above, we have $q^{\frac{k^2}{2}}$ possible choices for $A_{42}$.
    
    Therefore, we have $q^{kl} \times q^{\frac{k^2}{2}}\times q^{\frac{k^2}{2}}=q^{k^2+kl}=q^{k(k+ l)}= q^{\frac{kn}{2}}$ possible choices
    for the matrices $A_{31}, A_{41}$ and $A_{42}$. Therefore, the desired  result is proved. 
\end{proof}

The number of distinct Hermitian self-dual linear codes of even length $n$ over $\mathbb{F}_{q}+u\mathbb{F}_{q}+u^2\mathbb{F}_{q}$ can be summarized in the following theorem.

\begin{theorem}\label{NH3} Let $q$ be a square prime power and let $n$ be a positive integer.  Then  the number of distinct Hermitian self-dual linear codes of  length $n$ over $\mathbb{F}_{q}+u\mathbb{F}_{q}+u^2\mathbb{F}_{q}$ is 
    \[NH_3(q,n)= \begin{cases}
     \sigma_{\rm H}(q,n)\displaystyle\sum_{k=0}^{n/2}\begin{bmatrix}\displaystyle\dfrac{n}{2}\\ 
    k
    \end{bmatrix}_qq^{kn/2}   &\text{ if } n \textbf{ is even},
    \\
    0&\text{ otherwise} . \end{cases}
    \]
\end{theorem}

\section{Self-Dual Quasi-Abelian Codes  over Principal Ideal Group Algebras}
In this section, the study of quasi-abelian codes over  principal ideal group algebras is given. In the special case where the field characteristic is $3$ and the Sylow $3$-subgroup of the underlying finite abelian group has order $3$, complete characterization and enumeration of quasi-abelian codes  and self-dual quasi-abelian codes are presented in terms linear codes and self-dual linear  codes over $\mathbb{F}_{3^m}+u\mathbb{F}_{3^m}+u^2\mathbb{F}_{3^m}$  obtained  in \cite{BNV2018}, \cite{CG2015}  and Section  \ref{sec:3}.

\subsection{Group Rings and Quasi-Abelian Codes}
Let $R$ be a finite commutative ring with nonzero identity and let $G$ be a finite abelian group. Then
\begin{align*}
R[G]= \left\{ \displaystyle\sum_{g\in G } \alpha_gY^g \mid \alpha_g \in R, g\in G\right\}
\end{align*}
is a commutative ring under the addition and multiplication   given for  the usual polynomial ring over~$R$ with indeterminate~$Y$, where the indices are computed additively in~$G$.  The ring $R[G]$ is called a {\em group ring of~$G$ over~$R$}.  
In the case where~$R$ is the finite field~$\mathbb{F}_{p^m}$, the group ring $\mathbb{F}_{p^m}[G]$ is called a {\em group algebra of~$G$ over~$\mathbb{F}_{p^m}$} and it is called a {\em Principal Ideal Group Algebra (PIGA)} if every ideal  in $\mathbb{F}_{p^m}[G]$ is principal.  The readers may refer to  \cite{MCPSS2002} for more details on group rings. 
A linear code of length $|G|$ over $R$ can be  viewed as an  embedded  $R$-submodule of the $R$-module   in $R[G]$ by indexing the $|G|$-tuples by the elements in~$G$.
Given a subgroup $H$ of $G$ with index $n=[G:H]$, a linear code $\C$ of length $|G|$ viewed as an $R$-submodule of $R[G]$    is called an $H${\em -quasi-abelian code} (specifically, an $H${\em-quasi-abelian code of index $n$})  in $R[G]$ if $C$ is an $ R [ H ]$-module, i.e., $\C$ is closed under the multiplication by the elements in $R[H]$. Such a code will be called a quasi-abelian code if $H$ is not specified or where it is clear in the context.

Let $\{g_1,g_2,\ldots ,g_n\}$ be a fixed set of representatives of the cosets of $H$ in $G$. Let $\mathcal{R}:= \mathbb{F}_q[H]$. Define 􏰃$\Phi : \mathbb{F}_q[G]\to  \mathcal{R}^n$ by
\[\Phi \displaystyle\left(\displaystyle\sum_{h\in H}\sum_{i=1}^{n}\alpha_{h+g_i}Y^{h+g_i}\right)=(\alpha_1(Y),\alpha_2(Y),\ldots , \alpha_n(Y)),\]
where $\alpha_i(Y)=\displaystyle\sum_{h\in H}\alpha_{h+g_i}Y^h\in \mathcal{R}$ for all $i=1,2,\ldots, n$.
It is well known that 􏰃$\Phi$ is an $\mathcal{R}$-module isomorphism interpreted as follows.
\begin{lemma}[{\cite[Lemma 2.1]{JL2015}}]\label{relation}
    The map 􏰃$\Phi$ induces a one-to-one correspondence between $H$-quasi-abelian codes in $\mathbb{F}_q [G]$ and linear codes of length $n$ over $\mathcal{R}$.
\end{lemma}

We note that   a group algebra   $\mathbb{F}_{p^m}[H]$   is semisimple if and only if the Sylow $p$-subgroup of $H$ is trivial (see \cite[Chapter 2: Theorem 4.2]{Pa1977}), and  it   is    a PIGA   if and only if  he Sylow $p$-subgroup of $H$ is cyclic (see \cite{FiSe1976}).  
In \cite{JL2015}, complete characterization and enumeration of  $H$-quasi-abelian codes in  $\mathbb{F}_{p^m}[G]$  have been established in the case where $\mathbb{F}_{p^m}[H]$  is semisimple. Here, we focus on a more general case where $\mathbb{F}_{p^m}[H]$  is a PIGA, or equivalently,  the Sylow $p$-subgroup of $H$ is cyclic. Precisely, $H\equiv  A\times \mathbb{Z}_{p^mi} $ and  $G\cong A\times \mathbb{Z}_{p^s} \times B$, where $s$ is a non-negative integer,  $A$ and $B$  are finite abelian groups such that $p\nmid |A|$. General characterization is given in Subsection  4.2.  
In the special case where $p=3$ and $s=1$, complete characterization and enumeration of  $A\times \mathbb{Z}_3$-quasi-abelian codes and  self-dual $A\times \mathbb{Z}_3$-quasi-abelian codes   in $\mathbb{F}_{3^m}[A\times \mathbb{Z}_3\times B]$  are  given in Subsection 4.3.

\subsection{$A\times\mathbb{Z}_{p^s}$-Quasi-Abelian Codes  in $\mathbb{F}_{p^m}[A\times\mathbb{Z}_{p^s}\times B]$}

We focus on $H$-quasi-abelian codes in $\mathbb{F}_{p^m}[G]$, where  $\mathbb{F}_{p^m}[H]$  is a PIGA. Equivalently,    $H\equiv  A\times \mathbb{Z}_{p^s} $ and  $G\cong A\times \mathbb{Z}_{p^s} \times B$, where $s$ is a positive integer,  $A$ and $B$  are finite abelian groups such that $p\nmid |A|$ (see \cite{FiSe1976} and \cite{JLS2013}).

Note that the group algebra  $\mathbb{F}_{p^m}[A]$  is semisimple \cite{SD1967} and it can be decomposed using the Discrete Fourier Transform in \cite{RS1992} (see  \cite{JLS2013}  and \cite{JLLX2012} for more details). For completeness, the decomposition used in this paper are summarized as follows.

For co-prime positive integers $i$  and $j$, denote by ${\rm ord}_j(i)$ the multiplicative order of $i$ modulo $j$. For each $a\in A$, denote by ${\rm ord}(a)$ the additive order of $a$ in $A$ and the  $p^m${\em -cyclotomic class}   of $A$ containing $a\in A$  is defined to be the set
\begin{align*}
S_{p^m}(a)&:= \{p^{mi}\cdot a \mid i=0,1,\dots\}= \{p^{mi}\cdot a \mid 0\leq i< {\rm ord}_{{\rm ord}(a)}(p^m) \}, 
\end{align*}
where $p^{ki}\cdot a:= \sum\limits_{j=1}^{p^{mi}}a$ in $A$. A subset $\{a_1,a_2,\ldots ,a_t \}$ of~$A$ is called a {\em complete set of representatives} of $p^m$-cyclotomic classes of~$A$ if~$S_{p^m}(a_1), S_{p^m}(a_2), \ldots , S_{p^m}(a_t)$ are distinct and~$\displaystyle\bigcup_{i=1}^{t}S_{p^m}(a_i)=A$.

An {\em idempotent} in $\mathbb{F}_{p^m}[A]$ is a nonzero element $e$ such that $e^2=e$. It is called {\em primitive} if for every other idempotent $f$, either $ef=e$ or $ef=0$.  The existence of primitive idempotent elements in~$\mathbb{F}_{p^m}[A]$ is proved in \cite{DKL2000}.  They are induced by the $p^m$-cyclotomic classes of $A$ (see \cite[Proposition II.4]{DKL2000}). Consequently, $\mathbb{F}_{p^m}[A]$ can be viewed as a direct sum of principal ideals generated by these primitive idempotent elements.
\begin{proposition}[{\cite[Corollary III.6]{DKL2000}}]\label{decomposeR}Let $\{ a_1,a_2,\dots, a_t\}$ be a complete set of representatives of $p^m$-cyclotomic classes  of a finite abelian group~$A$ where $p\nmid |A|$ and let  $e_i$ be the primitive idempotent induced by $S_{p^m}(a_i)$  for all $1\leq i\leq t$. Then
    \[ \mathbb{F}_{p^m}[A]= \displaystyle\bigoplus_{i=1}^{t}\mathbb{F}_{p^m}[A]e_i 
    \cong \prod_{i=1}^{t} \mathbb{F}_{p^{m_i}},  \]
    where $m_i=m\cdot {\rm ord}_{{\rm ord}(a_i)}(p^m)$.
\end{proposition}

A PIGA $	\mathbb{F}_{p^m}[A\times\mathbb{Z}_{p^s}]  $ can be decomposed in the following theorem. 
\begin{theorem}\label{decomposeabelian0}Let   $s$ be a positive integer. Let $\{ a_1,a_2,\dots, a_t\}$ be a complete set of representatives of $p^m$-cyclotomic classes  of a finite abelian group~$A$ where $p\nmid |A|$.  
    Then
    \begin{align*}
    \mathbb{F}_{p^m}[A\times\mathbb{Z}_{p^s}]      \cong \displaystyle\prod_{i=1}^{t}\left(\mathbb{F}_{p^{m_i}}+u\mathbb{F}_{p^{m_i}}+\dots +u^{p^s-1}\mathbb{F}_{p^{m_i}}\right)
    \end{align*}
    where $m_i=m\cdot {\rm ord}_{{\rm ord}(a_i)}(p^m)$ for all $1\leq i\leq t$.
\end{theorem}
\begin{proof} For each $1\leq i\leq t$, let  $e_i$ be  the primitive idempotent induced by $S_{p^m}(a_i)$.  From Proposition~\ref{decomposeR}, we have 
    \begin{align}\label{decomab1}
    \mathbb{F}_{p^m}[A]
    \cong    	\mathbb{F}_{p^m}[A\times\mathbb{Z}_{p^s}] e_i \cong \prod_{i=1}^{t} \mathbb{F}_{p^{m_i}},
    \end{align}
    and hence, 
    \begin{align}\label{decomab2}
    \mathbb{F}_{p^m}[A\times\mathbb{Z}_{p^s}]  \cong  (	\mathbb{F}_{p^m}[A] )[\mathbb{Z}_{p^s}]\cong    \prod_{i=1}^{t} \mathbb{F}_{p^{m_i}} [\mathbb{Z}_{p^s}].
    \end{align}
    
    Under the ring isomorphism that fixes the elements~in $\mathbb{F}_{p^{m_i}}$ and $Y^1 \mapsto u+1,$ it is not difficult to see that 
    \begin{align}\label{decomab4}
    \mathbb{F}_{p^{m_i}}[\mathbb{Z}_{p^s}]\cong \mathbb{F}_{p^{m_i}}+u\mathbb{F}_{p^{m_i}}+\dots+u^{p^s-1}\mathbb{F}_{p^{m_i}}
    \end{align} as rings. Therefore,
    \begin{align}
    \mathbb{F}_{p^m}[A\times\mathbb{Z}_{p^s}]      \cong \displaystyle\prod_{i=1}^{t}\left(\mathbb{F}_{p^{m_i}}+u\mathbb{F}_{p^{m_i}}+\dots+u^{p^s-1}\mathbb{F}_{p^{m_i}}\right)
    \end{align} as desired. 
\end{proof}

For each finite abelian group $B$ of order $n$, every  $A\times \mathbb{Z}_{p^s}$-quasi-abelian code in  $\mathbb{F}_{p^m}[A\times\mathbb{Z}_{p^s}\times B]$ can be viewed as a linear code of length $n$ over  $\mathbb{F}_{p^m}[A\times\mathbb{Z}_{p^s}]  $  by  Lemma \ref{relation}.    The next corollary follows  directly from  Theorem~\ref{decomposeabelian0}.
\begin{corollary}\label{charabelian} Let $s$  and $m$ be positive integers. Let $A$ and $B$ be   finite abelian groups such that $|B|=n$ and   $p \nmid |A|$. Then every  $A\times \mathbb{Z}_{p^s}$-quasi-abelian code $\C$  in  $\mathbb{F}_{p^m}[A\times\mathbb{Z}_{p^s}\times B]$  can be viewed as
    \begin{align*}
    \C\cong  \prod_{i=1}^{t} \C_i, 
    \end{align*}
    where $\C_i$ is a linear code of length $n$ over $\mathbb{F}_{p^{m_i}}+u\mathbb{F}_{p^{m_i}}+\dots+u^{p^s-1}\mathbb{F}_{p^{m_i}}$ for all  $i=1,2,\dots,t$.
\end{corollary}

The enumeration of $A\times \mathbb{Z}_{p^s}$-quasi-abelian code in  $\mathbb{F}_{p^m}[A\times\mathbb{Z}_{p^s}\times B]$  is given as follows. 
\begin{theorem}\label{NA}Let $s$  and $m$ be positive integers. Let $A$ and $B$ be   finite abelian groups such that $|B|=n$ and  the exponent of $A$ is  $M$  and  $p \nmid M$. Then the number of $A\times \mathbb{Z}_{p^s}$-quasi-abelian codes in $\mathbb{F}_{p^m}[A\times\mathbb{Z}_{p^s}\times B]$ is 
    \begin{align*} 
    \prod_{d\mid M} \left(N_{p^s}( p^{m \cdot {\rm ord}_d(p^m)},n)\right)^{\frac{\mathcal{N}_A(d)}{{\rm ord}_d(p^m)}},
    \end{align*}
    where   $\mathcal{N}_A(d)$ is the number of elements of order~$d$ in~$A$  determined in \cite{B1997} and $N_{p^s}( p^{m \cdot {\rm ord}_d(p^m)},n)$ is the number of linear codes of length $n$ over $\mathbb{F}_{p^{m \cdot {\rm ord}_d(p^m)}}+u\mathbb{F}_{p^{m \cdot {\rm ord}_d(p^m)}}+\dots+u^{p^s-1}\mathbb{F}_{p^{m \cdot {\rm ord}_d(p^m)}}$ determined in {\cite[Lemma 2.2]{CG2015}}.
\end{theorem}
\begin{proof}
    From Theorem~\ref{decomposeabelian0}, it suffices to determine the number of linear codes of length~$n$ over the ring $ \mathbb{F}_{p^{m_i}}+u\mathbb{F}_{p^{m_i}}+\dots +u^{p^s-1}\mathbb{F}_{p^{m_i}}$  for all $i=1,2,\dots, t$.
    
    For each divisor $d$ of $M$, each  $p^m$-cyclotomic class containing an element of order $d$ has  
    ${{\rm ord}_d(p^m)}$  elements and the number of such  $p^m$-cyclotomic classes  is $\frac{\mathcal{N}_A(d)}{{\rm ord}_d(p^m)}$. By Theorem~\ref{decomposeabelian0}, it follows that the number of linear codes of length $n$  over $\mathbb{F}_{p^{m \cdot {\rm ord}_d(p^m)}}+u\mathbb{F}_{p^{m \cdot {\rm ord}_d(p^m)}}+\dots+u^{p^s-1}\mathbb{F}_{p^{m \cdot {\rm ord}_d(p^m)}}$  corresponding to $d$ is 
    \begin{align*}
    \left(N_{p^s}( p^{m\cdot{\rm ord}_d(p^m)},n)\right)^{ \frac{\mathcal{N}_A(d)}{{\rm ord}_d(p^m)}}.
    \end{align*}
    By taking the summation over all the divisors $d$ of $M$,    the desired result follows.
\end{proof}

In the next  subsections, we focus on self-dual $A\times\mathbb{Z}_{p^s}$-quasi-abelian codes  in $\mathbb{F}_{p^m}[A\times\mathbb{Z}_{p^s}\times B] $ with respect to both the Euclidean and Hermitian inner products.

\subsection{Euclidean Self-Dual $A\times\mathbb{Z}_{p^s}$-Quasi-Abelian Codes  in $\mathbb{F}_{p^m}[A\times\mathbb{Z}_{p^s}\times B]$}

Euclidean  self-dual $A\times\mathbb{Z}_{p^s}$-quasi-abelian codes  in $\mathbb{F}_{p^m}[A\times\mathbb{Z}_{p^s}\times B] $ is studied  in terms of  the following  types of $p^m$-cyclotomic classes. 
A $p^m$-cyclotomic class $S_{p^m}(a)$ is said to be of  {\em type} ${\I}$    if  $a=-a$ (in this case, $S_{p^m}(a)=S_{p^m}(-a)$), {\em type} ${\II}$  if $S_{p^m}(a)=S_{p^m}(-a)$ and $a\neq -a$, or {\em type} ${\III}$   if $S_{p^m}(-a)\neq S_{p^m}(a)$. 	The primitive idempotent $e$   {induced by} $S_ {p^m}(a)$   is said to be of type $\lambda\in\{\I,\II,\III\}$ if $S_{p^m}(a)$ is a $p^m$-cyclotomic class of type $\lambda$. 

Rearrange the terms in the decomposition in Theorem \ref{decomposeabelian0} based on the $p^m$-cyclotomic classes  of types $\I$, $\II$ and $\III$, we have the next theorem. 
\begin{theorem}Let $m $ and $s$  be   positive integers  and let  $A$ be a finite abelian group such that $ p\nmid |A|$. Then\label{decomposeabelian1}
    \begin{align*}
    \mathbb{F}_{p^m}[A\times\mathbb{Z}_{p^s}]  \cong  \left( \prod_{i=1}^{r_{\I}}\mathcal{R}_i\right) \times \left( \prod_{j=1}^{r_{\II}} \mathcal{S}_j\right) \times \left( \prod_{l=1}^{(r_{\III})/2} (\mathcal{T}_l\times \mathcal{T}_l)\right),
    \end{align*}
    where $r_{\I}, r_{\II}$ and $r_{\III}$ are the numbers of elements in a complete set of representatives of $p^m$-cyclotomic classes of $A$ of types $\I, \II$, and ${\III}$, respectively,   $  \mathcal{R}_i={F}_{p^{m}}+u\mathbb{F}_{p^{m}}+\dots+u^{p^s-1}\mathbb{F}_{p^{m}}$ for all $i=1,2,\dots, r_{\I}$, 	$ \mathcal{S}_j=\mathbb{F}_{p^{m_{r_{\I}+j}}}+u\mathbb{F}_{p^{m_{r_{\I}+j}}}+\dots +u^{p^s-1}\mathbb{F}_{p^{m_{r_{\I}+j}}}$ for all   $j=1,2,\dots, r_{\II}$,  and  $ \mathcal{T}_l =\mathbb{F}_{p^{m_{r_{\I}+r_{\II}+l}}}+u\mathbb{F}_{p^{m_{r_{\I}+r_{\II}+l}}} +\dots u^{p^s-1}\mathbb{F}_{p^{m_{r_{\I}+r_{\II}+l}}}  $ for all $l=1,2,\dots, {(r_{\III})/2}$.
\end{theorem}

Using Theorem \ref{decomposeabelian1} and the analysis similar to those in \cite[Section II.D]{JLLX2012}, a $A\times \mathbb{Z}_{p^s}$-quasi-abelian code $\C$  in $\mathbb{F}_{p^m}[A\times\mathbb{Z}_{p^s}\times B]$  and its Euclidean dual  are given.
\begin{proposition}\label{formabelianE} Let $s$  and $m$ be positive integers. Let $A$ and $B$ be   finite abelian groups such that $|B|=n$ and   $p \nmid |A|$. Then an  $A\times \mathbb{Z}_{p^s}$-quasi-abelian code in $\mathbb{F}_{p^m}[A\times\mathbb{Z}_{p^s}\times B]$ can be viewed as
    \begin{align}\label{Euclideanquasiabelian}
    \C\cong \left(\prod_{i=1}^{r_{\I}} \B_i  \right)\times \left(\prod_{j=1}^{r_{\II}} \C_j  \right)\times \left(\prod_{l=1}^{(r_{\III})/2}\left( \D_l\times \D_l^\prime\right) \right), 
    \end{align}
    where $\B_i$, $\C_j$, $\D_l$ and $\D_l^\prime$ are linear codes of length $n$ over   $\mathcal{R}_i$, $\mathcal{S}_j$, $\mathcal{T}_l$ and $\mathcal{T}_l$, respectively,  for all  $i=1,2,\dots,r_{\I}$, $j=1,2,\dots,r_{\II}$ and  $l=1,2,\dots,{(r_{\III})/2}$.
    
    Furthermore, the Euclidean dual of~$\C$ in~\eqref{Euclideanquasiabelian} is of the form 
    \begin{align*} 
    \C^{\perp_{\rm E}}\cong \left(\prod_{i=1}^{r_{\I}} \B_i^{\perp_{\rm E}}  \right) \times \left(\prod_{j=1}^{r_{\II}} \C_j ^{\perp_{\rm H}} \right) \times \left(\prod_{l=1}^{(r_{\III})/2}\left( (\D_l^\prime) ^{\perp_{\rm E}}\times  \D_l^{\perp_{\rm E}}\right) \right).
    \end{align*}
\end{proposition}
The characterization of Euclidean self-dual $A\times \mathbb{Z}_{p^s}$-quasi-abelian codes in $\mathbb{F}_{p^m}[A\times\mathbb{Z}_{p^s}\times B]$  is established in terms of a product of linear codes, Euclidean self-dual linear codes, and Hermitian self-dual linear codes over Galois extensions of the ring $\mathbb{F}_{p^m}+ u\mathbb{F}_{p^m}+\dots+u^{p^s-1}\mathbb{F}_{p^m}$.
\begin{corollary}\label{selfDuA} Let $s$  and $m$ be positive integers. Let $A$ and $B$ be   finite abelian groups such that $|B|=n$ and   $p \nmid |A|$. Then a $A\times \mathbb{Z}_{p^s}$-quasi-abelian code $C$ in $\mathbb{F}_{p^m}[A\times\mathbb{Z}_{p^s}\times B]$  is  Euclidean self-dual if and only if in the decomposition (\ref{Euclideanquasiabelian}), 
    \begin{enumerate}[$i)$]
        \item $\B_i$ is a Euclidean self-dual linear code of length $n$ over $\mathcal{R}_i$ for all $i=1,2,\dots,r_{\I}$,
        \item $\C_j$ is  a Hermitian  self-dual linear code of length $n$ over $\mathcal{S}_j$  for all $j=1,2,\dots,r_{\II}$, and       
        \item   $\D_l^\prime= \D_l^{\perp_{\rm E}}$ is a linear code of length $n$ over $\mathcal{T}_l$  for all $l=1,2,\dots,(r_{\III})/2$.
    \end{enumerate}
\end{corollary}
From Corollary \ref{selfDuA},  the Euclidean self-duality of $A\times \mathbb{Z}_{p^s}$-quasi-abelian code $\C$ in $\mathbb{F}_{p^m}[A\times\mathbb{Z}_{p^s}\times B]$ depends only on the structure of $A\times \mathbb{Z}_{p^s}$ and the index $n= |B|$ but not  the structure of $B$ itself.

Given positive integers $m$ and  $j$, the pair $(j,p^m)$  is said to be {\em  good} if $j$ divides $p^{mt}+1$ for some positive  integer $t$,   and {\em bad} otherwise.  This  notion  have been introduced in \cite{JLX2011} and \cite{JLLX2012} for  the enumeration of self-dual cyclic codes and self-dual abelian codes over finite fields and it is completely determined in \cite{J2018}. 
Let  $\chi$  be a function  defined on pairs $(j,p^m)$  as follows.
\begin{align}\label{chi}
\chi(j,p^m)
=\begin{cases}
0 &\text{ if } (j,p^m) \text{ is good},\\
1 &\text{ otherwise.}
\end{cases}
\end{align}

The number of Euclidean self-dual $A\times \mathbb{Z}_{p^s}$-quasi-abelian code $\C$ in $\mathbb{F}_{p^m}[A\times\mathbb{Z}_{p^s}\times B]$  can be determined as follows.
\begin{theorem} \label{NEA}
    Let $s$  and $m$ be positive integers. Let $A$ and $B$ be   finite abelian groups such that $|B|=n$ is even and  the exponent of $A$ is  $M$  and  $p \nmid M$. Then the number of   Euclidean self-dual $A\times \mathbb{Z}_{p^s}$-quasi-abelian codes  in $\mathbb{F}_{p^m}[A\times\mathbb{Z}_{p^s}\times B]$ is 
    \begin{align*} 
    \left(NE_{p^s}(p^m,n)\right)^{ \sum\limits_{{d\mid M, {\rm ord}_d(p^m)= 1}} (1-\chi(d,p^m))\mathcal{N}_A(d)}  &\times\prod_{\substack{d\mid M\\ {\rm ord}_d(p^m)\ne 1}} \left(NH_{p^s}(p^{m\cdot {\rm ord}_d(p^m)},n)\right)^{(1-\chi(d,p^m))\frac{\mathcal{N}_A(d)}{{\rm ord}_d(p^m)}}  \\
    &\times\prod_{d\mid M} \left(N_{p^s}( p^{m\cdot {\rm ord}_d(p^m)},n)\right)^{\chi(d,p^m)\frac{\mathcal{N}_A(d)}{2{\rm ord}_d(p^m)}} ,
    \end{align*}
    where  $\mathcal{N}_A(d)$ denotes the number of elements in $A$   of order $d$ determined in \cite{B1997}.
\end{theorem}
\begin{proof}  
    From Corollary~\ref{selfDuA}, it suffices  to determine the numbers of linear codes $\B_i$'s,   $\C_j$'s, and   $\D_l$'s such that $\B_i$ and   $\C_j$  are Euclidean and Hermitian self-dual, respectively.

    From \cite[Remark 2.5]{JLS2013}, the elements in $A$ of the same order are partitioned into $p^m$-cyclotomic classes of the same type.
    For each divisor $d$ of $M$, a  $p^m$-cyclotomic class containing an element of order $d$ has cardinality 
    ${{\rm ord}_d(p^m)}$  and the number of such  $p^m$-cyclotomic classes   is $\frac{\mathcal{N}_A(d)}{{\rm ord}_d(p^m)}$. We consider the following $3$ cases.
    
    \noindent {\bf  Case 1:}   $\chi(d,p^m)=0$  and  ${\rm ord}_d(3^k)=1$.   By \cite[Remark 2.6]{JLLX2012}, every  $3^k$-cyclotomic class    of $A$  containing an element  of order $d$  is of type  $\I$.  Since there are  $\frac{\mathcal{N}_A(d)}{{\rm ord}_d(p^m)}$ such $p^m$-cyclotomic classes,  the  number  of Euclidean self-dual linear codes $B_i$'s of length $n$ corresponding to $d$ is 
    \[\left(NE_{p^s}( p^m,n)\right)^{ \frac{\mathcal{N}_A(d)}{{\rm ord}_d(p^m)}}=\left(NE_{p^s}( p^m,n)\right)^{   (1-\chi(d,p^m))\mathcal{N}_A(d)}  .\]

    \noindent {\bf  Case 2:}   $\chi(d,p^m)=0$ and  ${\rm ord}_d(p^m)\ne 1$.
    By \cite[Remark 2.6]{JLLX2012}, every  $p^m$-cyclotomic class    of $A$  containing an element   of order $d$  is of type  $\II$ and of even cardinality ${\rm ord}_d(p^m)$.  Hence,  the  number  of Hermitian self-dual linear codes $C_j$'s of length $n$  corresponding to $d$ is 
    \begin{align*} \left(NH_{p^s}( p^{m\cdot {\rm ord}_d(p^m)} ,n)\right)^{ \frac{\mathcal{N}_A(d)}{{\rm ord}_d(p^m)}} =\left(NH_{p^s}( p^{m\cdot {\rm ord}_d(p^m)} ,n)\right)^{(1-\chi(d,p^m))\frac{\mathcal{N}_A(d)}{{\rm ord}_d(p^m)}}.
    \end{align*}
    
    \noindent {\bf  Case 3:}   $\chi(d,p^m)=1$.  By \cite[Lemma 4.5]{JLLX2012},  every  $p^m$-cyclotomic class    of $A$  containing an element   of order $d$  is of type  $\III$.  Then  the  number  of   linear codes $D_l$'s of length $n$  corresponding to $d$ is 
    \[\left(N_{p^s}( p^{m\cdot{\rm ord}_d(p^m)},n)\right)^{ \frac{\mathcal{N}_A(d)}{2{\rm ord}_d(p^m)}}=\left(N_{p^s}( p^{m\cdot{\rm ord}_d(p^m)},n)\right)^{\chi(d,p^m)\frac{\mathcal{N}_A(d)}{2{\rm ord}_d(p^m)}}.\]
    The formula for the number of   Euclidean self-dual $A\times \mathbb{Z}_{p^s}$-quasi-abelian codes in $\mathbb{F}_{p^m}[A\times\mathbb{Z}_{p^s}\times B]$  follows since $d$ runs over all divisors of $M$.
\end{proof}

\begin{remark}  In general,  the  numbers $NE_{p^s}( p^m,n)$ and $NH_{p^s}( p^m,n)$ in Theorem \ref{NEA} have not been well studied. In the case where the field characteristic is $3$, we have the following conclusions.
    
    \begin{enumerate}
        \item   The  numbers $N_3(3^m,n)$, $NE_{3}( 3^m,n)$ and $NH_{3}( 3^m,n)$ have been determined in Proposition \ref{N3}, \cite[Theorem 1]{BNV2018} and Theorem \ref{NH3}.   By Theorem \ref{NEA}, the   enumeration for  Euclidean self-dual $A\times \mathbb{Z}_3$-quasi-abelian codes in $\mathbb{F}_{3^m}[A\times\mathbb{Z}_3\times B]$   is  completed.
        .  
        
        \item The construction/characterization of  linear, Euclidean self-dual  and Hermitian self-dual   codes of length $n$ over $\mathbb{F}_{3^m}+u\mathbb{F}_{3^m}+u^2\mathbb{F}_{3^m}$ have been given in  \cite{BNV2018}, \cite{CG2015}  and in the proof of Proposition \ref{HSDC1}. Hence, the  construction/characterization  of Euclidean self-dual $A\times \mathbb{Z}_3$-quasi-abelian code in $\mathbb{F}_{3^m}[A\times\mathbb{Z}_3\times B]$ can be obtained {from Corollary \ref{selfDuA}.}
        
        \item Note that,  if $n$ is odd, there are no Hermitian self-dual linear  codes  of length $n$ over $\mathbb{F}_{3^m}+u\mathbb{F}_{3^m}+u^2\mathbb{F}_{3^m}$ by  {Theorem \ref{NH3}.} Hence,  there are no Euclidean self-dual $A\times \mathbb{Z}_3$-quasi-abelian codes in $\mathbb{F}_{3^m}[A\times\mathbb{Z}_3\times B]$ for all abelian groups $B$ of odd order.
        
    \end{enumerate}

\end{remark}

\subsection{Hermitian Self-Dual $A\times\mathbb{Z}_{p^s}$-Quasi-Abelian Codes  in $\mathbb{F}_{p^m}[A\times\mathbb{Z}_{p^s}\times B]$}
In this subsection, we focus on the case where $m$ is even and study Hermitian self-dual  $A\times \mathbb{Z}_{p^s}$-quasi-abelian codes in $\mathbb{F}_{p^m}[A\times\mathbb{Z}_{p^s}\times B]$.

The characterization and enumeration  of Hermitian self-dual  $A\times \mathbb{Z}_{p^s}$-quasi-abelian codes in $\mathbb{F}_{p^m}[A\times\mathbb{Z}_{p^s}\times B]$ are  given based on the decomposition of a  group algebra  $\mathbb{F}_{p^m}[A\times\mathbb{Z}_{p^s}]$ in terms of the following types of   $p^m$-cyclotomic classes of~$A$.  
A $p^m$-cyclotomic class $S_{p^m}(a)$ is said to be of  {\em type} $\Ip$    if    $S_{p^m}(a)=S_{p^m}(-p^{\frac{m}{2}}a)$ or {\em type} $\IIp$  if $S_{p^m}(a)\ne S_{p^m}(-p^{\frac{m}{2}} a)$. The primitive idempotent $e$   {induced by} $S_ {p^m}(a)$   is said to be of type $\lambda\in\{\Ip,\IIp\}$ if $S_{p^m}(a)$ is a $p^m$-cyclotomic class of type $\lambda$.

Rearrange the terms in the decomposition in Theorem \ref{decomposeabelian0} based on the $p^m$-cyclotomic classes  of types $\Ip$ and $\IIp$,   the next theorem follows. 
\begin{theorem}\label{decomposeabelian2}Let $m$ be an even positive  integer   and  let $A$ be a finite abelian group such that $p\nmid |A|$.
    \begin{align*}
    \mathbb{F}_{p^m}[A\times\mathbb{Z}_{p^s}]  \cong   \left( \prod_{j=1}^{r_{\Ip}} \mathcal{S}\right) \times \left( \prod_{l=1}^{(r_{\IIp})/2} (\mathcal{T}_l\times \mathcal{T}_l)\right),
    \end{align*}
    where $r_\Ip$ and $r_{\IIp}$ are the numbers of elements in a complete set of representatives of $p^m$-cyclotomic classes of $A$ of types $\Ip$ and ${\IIp}$, respectively,
    $ \mathcal{S}_j=\mathbb{F}_{p^{m_{j}}}+u\mathbb{F}_{p^{m_{j}}}+\dots+u^{p^s-1}\mathbb{F}_{p^{m_{j}}}$ for all   $j=1,2,\dots, r_{\Ip}$  and  $ \mathcal{T}_l =\mathbb{F}_{p^{m_{r_{\Ip}+l}}}+u\mathbb{F}_{p^{k_{r_{\Ip}+l}}}+\dots +u^{p^s-1}\mathbb{F}_{p^{m_{r_{\Ip}+l}}}$ for all $l=1,2,\dots, {(r_{\IIp})/2} $.
\end{theorem}
  
Using Theorem~\ref{decomposeabelian2} and the analysis similar to those in~\cite[Section II.D]{JLS2013}, the $A\times \mathbb{Z}_{p^s}$-quasi-abelian code $\C$ in $\mathbb{F}_{p^m}[A\times\mathbb{Z}_{p^s}\times B]$ and its  Hermitian dual  are given.
\begin{proposition}\label{formabelianH}
    Let $s$  and $m$ be positive integers such that $m$ is even. Let $A$ and $B$ be   finite abelian groups such that $|B|=n$ and   $p \nmid |A|$. Then an $A\times \mathbb{Z}_{p^s}$-quasi-abelian code in $\mathbb{F}_{p^m}[A\times\mathbb{Z}_{p^s}\times B]$ can be viewed as
    
    \begin{align}\label{Hermitianabelian} 
    \C\cong   \left(\prod_{j=1}^{r_{\Ip}} \C_j  \right)\times \left(\prod_{l=1}^{(r_{\IIp})/2} \left( \D_l\times \D_l^\prime\right) \right), \end{align}
    where $\C_j$, $\D_l$ and $\D_l^\prime$ are linear codes of length $n$ over   $\mathcal{S}_j$, $\mathcal{T}_l$ and $\mathcal{T}_l$, respectively,  for all  $j=1,2,\dots,r_{\Ip}$ and  $l=1,2,\dots,{(r_{\IIp})/2} $.
    
    Furthermore, the Hermitian dual of~$\C$ in~\eqref{Hermitianabelian} is of the form
    \begin{align*} 
    \C^{\perp_{\rm H}}\cong   \left(\prod_{j=1}^{r_{\Ip}} \C_j ^{\perp_{\rm H}} \right) \times \left(\prod_{l=1}^{(r_{\IIp})/2}  \left( (\D_l^\prime) ^{\perp_{\rm E}}\times  \D_l^{\perp_{\rm E}}\right) \right).
    \end{align*}
\end{proposition}
The characterization of Hermitian self-dual $A\times \mathbb{Z}_{p^s}$-quasi-abelian codes in $\mathbb{F}_{p^m}[A\times\mathbb{Z}_{p^s}\times B]$ in term of a product of linear codes, and Hermitian self-dual linear codes over Galois extensions of the ring $\mathbb{F}_{p^m}+ u\mathbb{F}_{p^m}+\dots+u^{p^s-1}\mathbb{F}_{p^m}$ is established.
\begin{corollary}\label{selfDuA2}
    
    Let $s$  and $m$ be positive integers such that $m$ is even. Let $A$ and $B$ be   finite abelian groups such that $|B|=n$ and   $p \nmid |A|$. Then an  $A\times \mathbb{Z}_{p^s}$-quasi-abelian code in $\mathbb{F}_{p^m}[A\times\mathbb{Z}_{p^s}\times B]$
    is Hermitian  self-dual if and only if in the decomposition (\ref{Hermitianabelian}), 
    \begin{enumerate}[$i)$]
        \item $\C_j$ is a  Hermitian self-dual linear code of length $n$ over $\mathcal{S}_j$ for all  $j=1,2,\dots,r_{\Ip}$, and       
        \item   $\D_l^\prime= \D_l^{\perp_{\rm E}}$  is a linear code of length $n$ over $\mathcal{T}_l$  for all $l=1,2,\dots,{(r_{\IIp})/2} $.
    \end{enumerate}
\end{corollary}
From Corollary \ref{selfDuA2}, it follows that the Hermitian self-duality of $A\times \mathbb{Z}_{p^s}$-quasi-abelian codes in $\mathbb{F}_{p^m}[A\times\mathbb{Z}_{p^s}\times B]$ depends only on the structure of $A\times \mathbb{Z}_{p^s}$ and the index $n= |B|$ but not  the structure of $B$ itself.

Given a positive integer $m$ and a positive integer $j$,  the pair $(j,p^m)$  is said to be {\em oddly good} if $j$ divides $p^{mt}+1$ for some odd positive integer $t$. This notion has been introduced in \cite{JLS2013}  for characterizing the Hermitian self-dual abelian codes in principal ideal group algebra and completely determined in \cite{J2018}. 

Let   $\lambda$ be a function defined on the pair $(j,p^m)$ as  
\begin{align}\label{lambda}
\lambda(j,p^m)=
\begin{cases}
0&\text{if }  (j,p^m) \text{ is oddly good},\\
1&\text{otherwise}.
\end{cases}
\end{align}

The number of Hermitian self-dual $A\times \mathbb{Z}_{p^s}$-quasi-abelian codes in $\mathbb{F}_{p^m}[A\times\mathbb{Z}_{p^s}\times B]$ can be determined as follows.
\begin{theorem} \label{NHA}  Let $s$  and $m$ be positive integers such that $m$ is even. Let $A$ and $B$ be   finite abelian groups such that $|B|=n$ is even and  the exponent of $A$ is  $M$  and  $p \nmid M$. Then the number of   Euclidean self-dual $A\times \mathbb{Z}_{p^s}$-quasi-abelian codes  in $\mathbb{F}_{p^m}[A\times\mathbb{Z}_{p^s}\times B]$ is  
    \begin{align*} 
    \prod_{{d\mid M}} \left(NH_{p^s}( p^{m\cdot{\rm ord}_d(p^m)}, p^s)\right)^{(1-\lambda(d,p^\frac{m}{2}))\frac{\mathcal{N}_A(d)}{{\rm ord}_d(p^m)}}  \times\prod_{d\mid M} \left(N_{p^s}( p^{m\cdot{\rm ord}_d(p^m)},p^s)\right)^{\lambda(d,p^\frac{m}{2})\frac{\mathcal{N}_A(d)}{2{\rm ord}_d(p^m)}}, 
    \end{align*}
    where  $\mathcal{N}_A(d)$ denotes the number of elements of order $d$ in $A$ determined in \cite{B1997}.
\end{theorem}
\begin{proof}
    By Corollary~\ref{selfDuA2}, it is enough to determine the numbers linear codes $\C_j$'s and $\D_l$'s of length $n$  in (\ref{Hermitianabelian}) such that $\C_j$ is Hermitian self-dual.    
    The   result  can be deduced  using arguments  similar to those in the proof of  Theorem \ref{NEA}, where \cite[Lemma 3.5]{JLS2013} is applied instead of  \cite[Lemma 4.5]{JLLX2012}.
\end{proof}

\begin{remark} In general,  the number   $NH_{p^s}( p^m,n)$ of Hermitian self-dual linear codes of length $n$ over $\mathbb{F}_{p^m}+u\mathbb{F}_{p^m}+\dots+u^{p^s-1}\mathbb{F}_{p^m}$  in Theorem \ref{NHA} has not been well studied. 
    In the case where the field characteristic is $3$,  we have the following results.
    \begin{enumerate}
        \item  The  numbers  $N_{3}( p^m,n)$ and $NH_{3}( 3^m,n)$ have been determined in Proposition \ref{N3}  and Theorem \ref{NH3}.   Hence,  the  complete enumeration of  Hermitian self-dual $A\times \mathbb{Z}_3$-quasi-abelian codes in $\mathbb{F}_{3^m}[A\times\mathbb{Z}_3\times B]$ follows. 
        \item The construction/characterization of  linear and Hermitian self-dual dual linear codes of length $n$ over $\mathbb{F}_{3^m}+u\mathbb{F}_{3^m}+u^2\mathbb{F}_{3^m}$ have been given in \cite{CG2015} and in the proof of Proposition \ref{HSDC1}. Hence, the  construction/characterization  of Hermitian self-dual $A\times \mathbb{Z}_3$-quasi-abelian code in $\mathbb{F}_{3^m}[A\times\mathbb{Z}_3\times B]$ can be obtained from Corollary \ref{selfDuA2}.
        
        \item Note that,  if $n$ is odd, there are no Hermitian self-dual linear  codes  of length $n$ over $\mathbb{F}_{3^m}+u\mathbb{F}_{3^m}+u^2\mathbb{F}_{3^m}$ by {Theorem \ref{NH3}.} Hence,  there are no Hermitian self-dual $A\times \mathbb{Z}_3$-quasi-abelian codes in $\mathbb{F}_{3^m}[A\times\mathbb{Z}_3\times B]$ for all abelian groups $B$ of odd order.
        
    \end{enumerate}
    
\end{remark}

\section{Conclusion and Remarks}
By extending the technique used in the study of  Euclidean self-dual linear codes  over $\mathbb{F}_{q}+u\mathbb{F}_{q}+u^2\mathbb{F}_{q}$ in  \cite{BNV2018},  complete characterization and enumeration    of Hermitian self-dual linear codes  over $\mathbb{F}_{q}+u\mathbb{F}_{q}+u^2\mathbb{F}_{q}$  have been established for all square prime powers $q$.  Subsequently,  algebraic characterization  of $H$-quasi-abelian codes  in $\mathbb{F}_{p^m}[G]$ has been  studied, where  $H\leq G$ are finite abelian groups and the Sylow $p$-subgroup of $H$ is cyclic, or equivalently,  $\mathbb{F}_q[H]$ is a principal ideal group algebra.  In the special case where $H\cong  A\times \mathbb{Z}_3$ with $3\nmid |A|$,  characterization and enumeration of  $H$-quasi-abelian codes and  self-dual $H$-quasi-abelian codes   in $\mathbb{F}_{3^m}[H\times B]$   have been completely determined for all finite abelian group $B$.   As applications, characterization and enumeration of  self-dual $A\times \mathbb{Z}_3$-quasi-abelian codes in~$\mathbb{F}_{3^m}[A\times \mathbb{Z}_3\times B]$  can be  presented  in terms of  linear codes and  self-dual linear codes over  some extensions of   $\mathbb{F}_{3^m}+u\mathbb{F}_{3^m}+u^2\mathbb{F}_{3^m}$ determined in \cite{BNV2018}, \cite{CG2015}  and Section \ref{sec:3}.

In general, it would be interesting  to  studied   $A\times P$-quasi-abelian codes and self-dual  $A\times P$-quasi-abelian codes in $\mathbb{F}_{p^m}[A\times P\times B]$ for  all primes $p$ and finite abelian $p$-groups $P$. For $e\geq 4$, characterization and enumeration of self-dual linear codes 
over $\mathbb{F}_{p^{m}}+u\mathbb{F}_{p^{m}}+\dots +u^{e-1}\mathbb{F}_{p^{m}}$ are other interesting  problems.



\end{document}